\documentclass{article}

 \PassOptionsToPackage{numbers, compress}{natbib}

\usepackage[preprint]{neurips_2021}
 
\usepackage[utf8]{inputenc} 
\usepackage[T1]{fontenc}    
\usepackage{hyperref}       
\usepackage{url}            
\usepackage{booktabs}       
\usepackage{amsfonts}       
\usepackage{nicefrac}       
\usepackage{microtype}      
\usepackage{xcolor}         

\usepackage{times}

\usepackage{soul}
\usepackage{url}
\usepackage[small]{caption}
\usepackage{graphicx}
\usepackage{amsmath}
\usepackage{mathtools}
\usepackage{booktabs}
\usepackage{amsfonts}
\usepackage{amsthm}
\usepackage{amssymb}
\usepackage{bbm}
\usepackage{subcaption}
\usepackage{bm}
\usepackage[ruled,linesnumbered]{algorithm2e}
\usepackage{algorithmic}
\usepackage{color}

\usepackage{multirow}

\newtheorem{example}{Example}
\newtheorem{theorem}{Theorem}

\newtheorem{lemma}[theorem]{Lemma}
\newtheorem{corollary}[theorem]{Corollary}

\newtheorem{definition}{Definition}
\newtheorem{statement}[theorem]{Statement}

\newcommand{\vx}{\vec x}
\newcommand{\xx}{\ell}
\newcommand{\vzero}{\vec 0}

\newcommand{\vv}{\vec v}

\newcommand{\HG}{\text{\rm Hist}}

\newcommand{\vA}{\bm A}

\renewcommand{\vv}{\vec v}
\newcommand{\vell}{\vec \ell}
\newcommand{\vb}{\vec b}
\newcommand{\vc}{\vec c}

\newcommand{\vP}{ P}
\newcommand{\vPM}{\vec{\omega}}
\newcommand{\vpi}{\vec{\pi}}
\newcommand{\DP}{\text{\rm{DP}}}
\newcommand{\true}{\text{\rm{true}}}
\newcommand{\false}{\text{\rm{false}}}

\newcommand{\PM}{\omega}
\newcommand{\T}{\mathsf{T}}
\newcommand{\pp}{\phi}

\newcommand{\bbR}{\mathbb{R}}
\newcommand{\bbZ}{\mathbb{Z}}
\newcommand{\CH}{\text{CH}}
\newcommand{\calA}{\mathcal{Q}}

\newcommand{\calT}{\mathcal{E}}

\newcommand{\calC}{\mathcal{C}}
\newcommand{\calH}{\mathcal{H}}

\newcommand{\sign}{\text{sign}}

\newcommand{\lnorm}{\left|\left|}
\newcommand{\rnorm}{\right|\right|}

\renewcommand{\sign}{\text{sign}}
\newcommand{\bbone}{\mathbbm {1}}
\newcommand{\Hist}{\text {Hist}}

\renewcommand{\dim}{\text {dim}}

\newcommand{\CL}{\pp}

\newcommand{\iagent}{j}
\newcommand{\numpm}{p}
\newcommand{\numagent}{n}
\newcommand{\ipm}{i}
\newcommand{\npm}{\numpm}
\newcommand{\nagent}{\numagent}

\newcommand{\setvote}{\mathcal{V}}
\newcommand{\setidx}{\bm{\Omega}}
\newcommand{\nvote}{m}
\newcommand{\hPM}{{\alpha}}
\newcommand{\hvPM}{{\vec{\alpha}}}

\newcommand{\prob}{\text{\rm{pr}}}
\newcommand{\threshold}{q}
\newcommand{\breaking}{d}
\newcommand{\vthreshold}{\vec \threshold}
\newcommand{\vbreaking}{\vec \breaking}
\newcommand{\setfrac}{\setvote_{\text{\rm frc}}}

\newcommand{\pmax}{\widetilde{\DP}_{\Pi,r_{\vthreshold,\vbreaking},f}^{\max}(\nagent)}
\newcommand{\pmin}{\widetilde{\DP}_{\Pi,r_{\vthreshold,\vbreaking},f}^{\min}(\nagent)}

\newcommand{\pmaxg}{\widetilde{\DP}_{\Pi,r,f}^{\max}(\nagent)}
\newcommand{\pming}{\widetilde{\DP}_{\Pi,r,f}^{\min}(\nagent)}

\title{The Smoothed Likelihood of Doctrinal Paradox}

\author{%
  Ao Liu and Lirong Xia \\
  Department of Computer Science, Rensselaer Polytechnic Institute\\
  \texttt{liua6@rpi.edu} and \texttt{xial@cs.rpi.edu}
}

\begin{document}

\maketitle
\begin{abstract}
When aggregating logically interconnected judgements from $\nagent$ agents, the result might be logically inconsistent. 
This phenomenon is known as the \emph{doctrinal paradox}
, which plays a central role in the field of judgement aggregation. 
Previous work has mostly focused on the  worst-case analysis of the doctrinal paradox, leading to many impossibility results. Little is known about its likelihood of occurrence in practical settings,  except for the study under   certain~distributions~\citep{List2005:The-probability}. 

In this paper, we characterize the likelihood of the doctrinal paradox under a  general and realistic model called {\em smoothed social choice framework} introduced in a NeurIPS 2020 paper~\citep{Xia2020:The-Smoothed}. In the framework, agents' ground truth judgements can be arbitrarily correlated, while the noises are independent.  Our main theorem states that under mild conditions, the smoothed likelihood of the doctrinal paradox is either $0$, $\exp(-\Theta(\nagent))$, $\Theta(n^{-1/2})$ or $\Theta(1)$. This not only answers open questions by~\citet{List2005:The-probability} in 2005, but also draws clear lines between situations with frequent paradoxes and with vanishing paradoxes. 

\end{abstract}

\section{Introduction}
The field of {\em judgement aggregation} studies the aggregation of agents' judgements (yes/no votes) on multiple logically interconnected propositions in a consistent way. For example, suppose a defendant is involved in a traffic accident where a victim is injured. A jury of $\nagent$  agents (jurors) are asked to make collective decisions on the following three propositions:


\begin{center}
\begin{tabular}{l}
\noindent \emph{Proposition $\PM_1$}: 
whether the injury is caused by the defendant.\\
\emph{Proposition $\PM_2$}: whether the injury is serious enough. \\
\emph{Proposition $\PM_3$}: whether the defendant is guilty.
\end{tabular}
\end{center}
Each agent is asked to provide a yes/no judgement to each proposition in a logically consistent way, i.e., her answer to $\PM_3$ is Y if and only if her answers to both $\PM_1$ and $\PM_2$ are Y, or equivalently, $\PM_3 {\leftrightarrow} \PM_1\wedge\PM_2$ must hold. Then, agents' judgements, called a {\em profile}, are aggregated by a  rule $r$. 

One classical rule for judgement aggregation is {\em proposition-wise majority}, which performs majority voting on each proposition separately. Unfortunately, the output of proposition-wise majority can be logically inconsistent. This phenomenon is known as the {\em doctrinal paradox},   as shown in the following example. 



\begin{minipage}[t][][b]{0.45 \textwidth}
\begin{example}[\bf Doctrinal paradox]
\label{eg:DP} Suppose there are three three agents ($n=3$), whose judgements  
on $\PM_1$, $\PM_2$, and $\PM_3$ are shown in Table~\ref{tab:dp_intro}.  Notice that while every agent's judgements are consistent with the logic relationship $\PM_3 {\leftrightarrow} \PM_1\wedge\PM_2$,  the output of proposition-wise majority is inconsistent with the logic.
\end{example}
\end{minipage}
\hfill 
\begin{minipage}[t][][b]{0.5\textwidth}
    \begin{tabular}[t]{@{\ }r@{\ }@{\ }c@{\ }@{\ }c@{\ }@{\ }c@{\ }@{\ }|c}
    \hline
     Propositions & $\PM_1$ & $\PM_2$ & $\PM_3$& $ \PM_3 \overset{?}{\leftrightarrow} \PM_1\wedge\PM_2$ \\\hline
     Agent 1   & Y & N  & N & \textbf{True}\\
     Agent 2   & N & Y& N & \textbf{True}\\
     Agent 3   & Y& Y & Y & \textbf{True}\\\hline
     Aggregation   & Y & Y & N & \textbf{False}\\
    \hline
    \end{tabular}
    \captionof{table}{\small  Doctrinal paradox.\label{tab:dp_intro}}

\end{minipage}


In general, the doctrinal paradox can happen for $p\ge 1$ {\em premises} (logically independent propositions), plus one {\em conclusion}, whose value is determined by the values of the premises. The doctrinal paradox {\em ``originated the whole field of judgement aggregation''}~\citep{Grossi2014:Judgment}, and continued to play a central role since then~\citep{List2012:The-theory,Grossi2014:Judgment,Brandt2016:Handbook}. A large body of literature has shown the negative effects of the doctrinal paradox in  law~\citep{Kornhauser1986:Unpacking,Hanna2009:The-paradox,Chilton2012:The-doctrinal}, economics~\citep{Mongin2019:The-Present}, computational social choice~\citep{Bonnefon2010:Behavioral,Brandt2016:Handbook}, philosophy~\citep{Sorensen2003:A-Brief,Mongin2012:The-doctrinal} and psychology \citep{Bonnefon2011:The-doctrinal}, etc.  
Unfortunately, the paradox is unavoidable under mild assumptions~\citep{List2002:Aggregating}, and there is a large body of literature on impossibility theorems as well as ways for escaping from them~\citep{List2010:Introduction,Grossi2014:Judgment,Endriss2016:Judgment}. 

Despite the mathematical impossibilities, it is unclear whether the doctrinal paradox is a significant concern in practice. Empirical studies have been limited by the lack of datasets, as many classical applications of judgement aggregations are in high-stakes, low frequent, privacy-sensitive domains, such as judicial systems. To embrace AI-powered e-democracy, we believe that it is now the right time to theoretically study the likelihood of the paradox under realistic models. 
A low likelihood is desirable, because it suggests that, despite the fact that the paradox exists in the worst case, it is likeliho not be a big concern in practice. While this approach is highly popular in voting~\citep{Nurmi1999:Voting,Gehrlein2011:Voting,Diss2021:Evaluating}, it has not received due attention in judgement aggregation.  The only known exception is the paper by~\citet{List2005:The-probability}, who provided sufficient conditions for the likelihood of doctrinal paradox to converge to $1$ and to $0$, respectively, under some distributions for two or three premises.  Therefore, the following research question remains largely open:

\vspace{0.4em}
\centerline{\textbf{\noindent\em How likely does the doctrinal paradox occur under realistic models?}}

Successfully addressing this question faces two challenges. The first challenge is {\em modeling}, i.e., it is unclear what probabilistic models are realistic. A large body of literature in social choice focused on the  i.i.d.~uniform distribution, known as {\em Impartial Culture (IC)}, which has been widely criticized of being unrealistic~\citep{Nurmi1999:Voting}. The second challenge is the {\em technical hardness} in mathematically analyzing the likelihood. In fact, even characterizing the likelihood of doctrinal paradox under i.i.d.~distributions is already highly challenging, and~\citeauthor{List2005:The-probability}'s 2015 work~\citep{List2005:The-probability} left many i.i.d.~distributions as open questions, including IC. 

{\noindent \bf Our Model.  }We adopt the \emph{smoothed social choice} framework~\citep{Xia2020:The-Smoothed} to address the modeling challenge. The framework was inspired by the celebrated smoothed complexity analysis~\citep{Spielman2009:Smoothed}, and is more realistic and general than i.i.d.~distributions, especially IC.   
We consider the setting with $p\ge 1$ premises and one conclusion determined by the logical connection function $f$, where every agent's vote is a random variable whose distribution is in a set $\Pi$ over all $2^p$ judgements on premises. Then, given the aggregation rule $r$, a max-adversary aims at maxing the likelihood of doctrinal paradox by setting agents' distributions from $\Pi$. This leads to the definition of the {\em max-smoothed likelihood} of doctrinal paradox as follows.
\begin{equation}\nonumber
\pmaxg\triangleq \sup\nolimits_{\vpi\in\Pi^{\nagent}}{\Pr}_{\vP\sim\vpi}\left(\vP\text{ is a doctrinal paradox}\right),
\end{equation}
where $\vpi = \left(\pi_1,\cdots,\pi_{\nagent}\right)\in\Pi^n$ is the collection of the distributions of $\nagent$ agents chosen by the adversary, and profile $\vP$ is the collection of $\nagent$ agents' votes generated from $\vpi$. Similarly,  the {\em min-smoothed likelihood} of doctrinal paradox is defined as follows.
\begin{equation}\nonumber
\pming\triangleq \inf\nolimits_{\vpi\in\Pi^{\nagent}}{\Pr}_{\vP\sim\vpi}\left(\vP\text{ is a doctrinal paradox}\right).
\end{equation}
Notice that in both definitions, agents' ``ground truth'' preferences, represented by  $\vec\pi$, are arbitrarily correlated. The max-(respectively, min-) smoothed likelihood corresponds to the worst-case (best-case) analysis. A low max-smoothed likelihood is positive news, because it states that the paradox is unlikely to happen regardless of the adversary's choice. A high min-smoothed likelihood is strongly negative news, because it states that the paradox is likely to happen with non-negligible probability, regardless of agents' ground truth preferences.


{
\noindent \bf Our Contributions. } In this paper, we prove the following general characterization of the smoothed likelihood of doctrinal paradox.

{\vspace{2mm}\noindent \bf Theorem~\ref{theo:main}. }{\bf (Smoothed Likelihood of Doctrinal Paradox, Informal.)} \emph{Under mild assumptions, for any $\nagent\in\mathbb N$, any quota rule $r$, and any logical connection $f$, $\pmaxg$ is either $0$, $\exp(-\Theta(\nagent))$, $\Theta(n^{-1/2})$ or $\Theta(1)$, and $\pming$ is either $0$, $\exp(-\Theta(\nagent))$, $\Theta(n^{-1/2})$ or $\Theta(1)$. 
}



{\vspace{2mm}} The formal statement of the theorem also characterizes the condition for each case to happen. As commented above, the first three cases ($0$, $\exp(-\Theta(\nagent))$, and $\Theta(n^{-1/2})$) of $\pmaxg$ are positive news, because they state that the doctrinal paradox vanishes as $n\rightarrow\infty$, see \emph{e.g.}, Example~\ref{eg:1}. The last case ($\Theta(1)$) of $\pming$ is negative news. 
Notably, a direct corollary of Theorem~\ref{theo:main} (Corollary~\ref{coro:main}) characterizes the likelihood of doctrinal paradox under i.i.d.~distributions that are not covered by~\citet{List2005:The-probability}, in particular the uniform distribution which naturally corresponds to IC in social choice. 
Experiments on synthetic data in Section~\ref{sec:exp} confirm our theoretical results. 

\noindent{\bf Proof techniques.} The proof of Theorem~\ref{theo:main} addresses technical hardness in~\citep{List2005:The-probability} by first converting the smoothed doctrinal paradox to a union of polyhedra, and then applying~\citep[Theorem~2]{Xia2021:How-Likely}. Such applications can be highly non-trivial as commented in~\citep{Xia2021:How-Likely}, which we believe to be the case of this paper and are our main technical contribution.

{\vspace{2mm}\noindent \bf Related Work and Discussions. } Our work extends the results  by~\citet{List2005:The-probability} on i.i.d.~distributions in three dimensions. First, we characterize the smoothed likelihood of doctrinal paradox, which is more general and realistic than i.i.d.~distributions. 
Second, our theorem works for arbitrary number of premises and arbitrary logic connection function, while \citet{List2005:The-probability}'s results only work for two or three premises and conjunctive and disjunctive functions. Third, our theorem works for arbitrary quota rules, which extends the proposition-wise majority rule in~\citep{List2005:The-probability}. As discussed above, as our theorem is a full characterization, a straightforward corollary of our theorem addresses the likelihood of doctrinal paradox under all i.i.d.~cases that are left open in~\citep{List2005:The-probability}. 
 
The modern formulation of doctrinal paradox was due to~\citet{Kornhauser1986:Unpacking}, though a similar paradox was described by~\citeauthor{Poisson1837:Recherches} in 1937~\citep{Poisson1837:Recherches} and  by~\citeauthor{Vacca1921:Opinioni} in 1921~\citep{Vacca1921:Opinioni}, see~\citep{Grossi2014:Judgment}. The paradox is closely related to the more general {\em discursive dilemma}~\citep{Pettit2001:Deliberative}, where agents also ``vote'' on the logic relationship as in the last column in Table~\ref{tab:dp_intro}. \citet{List2002:Aggregating} proved the first {impossible theorem} about doctrinal paradox, which has lead to many followup impossibility theorems~\citep{Pauly2006:Logical,Mongin2008:Factoring,Dietrich2008:Judgment,Awad2017:Judgement,Mongin2019:The-Present,Eyal-Baharad2020:The-rarity,Marcoci2020:Judgement}. Methods to escape from the paradox have been explored~\citep{Nehring2021:The-median,Rahwan2010:Collective,Nehring2008:Consistent,Lyon2013:The-wisdom}. The connection between judgement aggregation and belief merging was established~\citep{Everaere2017:An-Introduction}. Computational and machine learning aspects of judgement aggregation have also been explored, see, e.g.,~\citep{Endriss2012:Complexity,Zhang2019:A-PAC-Framework,Endriss2019:The-Complexity,Baumeister2020:Complexity}. 



\section{Preliminaries}\label{sec:prelim}
{\noindent \bf   Basics of Judgement Aggregation. }
Let $[\nagent] \triangleq \{1,\cdots,\nagent\}$ denote the set of $\nagent$ agents. Let $\numpm$ denote the number of premises. The binary judgement of agent $\iagent$ on the $\ipm$-th premise is denoted by $\PM_{\iagent,\ipm}\in\{0,1\}$, where $\PM_{\iagent,\ipm}=1$ means ``Yes'' and $\PM_{\iagent,\ipm}=0$ means ``No''. The conclusion is often denoted by $\CL$, and sometimes by $\PM_{\npm+1}$ for simplicity. Because all premises are binary, there are $\nvote=2^{\npm}$ different combinations of {\em judgements} on the premises. 
To better present the results, instead of asking the agents to submit their judgements $\vPM_{\iagent} = (\PM_{\iagent,1},\cdots,\PM_{\iagent,\npm})$, equivalently, we ask each agent to submit an $\nvote$-dimensional vector $\vv_\iagent$, called her \emph{vote}, whose $\vPM_{\iagent}$-th component is $1$ and all other components are $0$'s. Let $\setvote \triangleq \left\{\vv\in\{0,1\}^{\nvote}:\nvote = 2^\numpm\text{ and }\lnorm\vv\rnorm_1=1\right\}$ denote the set of all  votes.  For example, in the setting of Example~\ref{eg:DP}, we have $p=2$, $m=4$, and the four judgements are $((0,0),(0,1),(1,0),(1,1))$. Agent $1$'s judgements are $(1,0)$, which means that her vote is $(0,0,1,0)$, where the third component corresponds to the weight on judgements $(1,0)$.

A probability distribution  over the $2^\numpm$ judgements is also called a {\em fractional vote}. Let $\setfrac \triangleq \left\{\vv\in[0,1]^{\nvote}:\lnorm\vv\rnorm_1=1\right\}$ denote the set of all fractional votes. Clearly,  $\setvote\subseteq \setfrac$.  For any fractional vote $\vv$ and any judgement $\vPM$,  we let $\vv(\vPM)$ denote the weight (probability) on $\vPM$. For example, agent $2$'s vote in Table~\ref{tab:eg2} below is a fractional vote with $0.5$ weight on each of $(1,0)$ and $(1,1)$. 


Let  $f:\{0,1\}^{\npm}\to\{0,1\}$ denote  the {\em logical connection} between the premises and the conclusion. That is, for any judgements $\vPM\in\{0,1\}^{\npm}$, $f(\vPM)$ is the conclusion that is consistent with the logic.  The domain of $f$ can be naturally extended to  (non-fractional) votes $\setvote$. That is, 
\begin{equation}\nonumber
\forall\,\iagent\in[\nagent],\;  \CL_{\iagent} = f(\vv_{\iagent}) = f(\vPM),\text{ where }\vv(\vPM)=1.
\end{equation}
For example, in the setting of Example~\ref{eg:DP}, the conclusion is $1$ if and only if both premises are $1$. Therefore,  $
f(\vv)= \left\{\begin{array}{lcl}
1      &      & \text{if } \vv = (0,\,0,\,0,\,1)\\
0      &      & \text{otherwise}
\end{array}\right..$

For any $1\le i\le p+1$, let $\setidx_{\ipm}$ denote the set of all judgements whose $i$-th proposition is $1$. That is,
\begin{equation}\nonumber
\begin{split}
\setidx_{\ipm} \triangleq \left\{\begin{array}{ll}
\big\{\vPM\in\{0,1\}^{\npm}:\;\,\vPM(\ipm)=1\big\}      &       \forall \ipm\in[\npm]\\
\big\{\vPM\in\{0,1\}^{\npm}:\,f(\vPM)=1\big\}      &       \text{if }\ipm = \npm+1
\end{array}\right.. 
\end{split}
\end{equation}
For example, in the setting of Example~\ref{eg:DP}, $\setidx_{1} = \{(1,0), (1,1)\}$, $\setidx_{2} = \{(0,1), (1,1)\}$, and $\setidx_{3} = \{(1,1)\}$.

For any $n\in\mathbb N$, let $\vP = (\vv_{1},\cdots,\vv_{\nagent})\in(\setfrac)^{\nagent}$ denote a {\em (fractional) profile} of $n$ agents.
A {\em (judgement) aggregation rule} is a function $r:(\setfrac)^{\nagent}\to\{0,1\}^{\npm+1}$, which takes a profile as input and outputs binary values for all premises and the conclusion.   For any (fractional) profile $\vP\in(\setfrac)^{\nagent}$, we define $\HG(\vP) \triangleq \sum_{\iagent\in[\nagent]}\vv_{\iagent}$  to be the {\em histogram} of $P$, which represents the total weight of each combination of judgements in $P$. We define  fractional votes and fractional profiles to present conditions in the theorem. Notice that agents' votes are in $\setvote$, which means that they are  non-fractional.  

{ \noindent \bf Quota Rules. } A quota rule is a natural and well-studied generalization of proposition-wise majority rule with different thresholds.  
Formally, given any vector of \emph{acceptance thresholds} (\emph{threshold} in short), denoted by $\vthreshold = (\threshold_1,\cdots,\threshold_{\npm+1})\in[0,1]^{\npm+1}$ and any vector of \emph{tie-breaking criteria}  (\emph{breaking} in short), denoted by  $\vbreaking = (\breaking_1,\cdots,\breaking_{\npm+1})\in\{0,1\}^{\npm+1}$, we define the quota rule $r_{\vthreshold,\vbreaking}(\vP)$ as follows. For any  profile $\vP\in(\setfrac)^{\nagent}$ and any $i\in [\npm+1]$, we let $\nagent_{\ipm} = \sum_{\vPM\in\setidx_i}\HG(\vP)(\vPM)$ denote the total weight of agents whose $i$-th judgement  is $1$, then apply the quota rule with threshold $q_i$ and breaking $d_i$. That is,
\begin{equation}\label{equ:quota}
\forall\;\ipm\in[\npm+1],\;r_{\vthreshold,\vbreaking}(\vP)(\ipm) \triangleq \left\{\begin{array}{lcl}
1      &      & \text{if }  \nagent_\ipm > \threshold_{\ipm}\cdot n\\
\breaking_\ipm      &      & \text{if }  \nagent_\ipm = \threshold_{\ipm}\cdot n\\
0      &      & \text{otherwise}
\end{array}\right.,
\end{equation}
where $r_{\vthreshold,\vbreaking}(\vP)(\ipm)$ is the $\ipm$-th component of $r_{\vthreshold,\vbreaking}(\vP)$. We say that a profile $\vP$ is {\em tied in $\PM_{\ipm}$} if   $\nagent_\ipm = \threshold_\ipm\cdot\nagent$.


{\noindent \bf   Doctrinal Paradox. } A profile $\vP$ is said to be a doctrinal paradox  under   $r$, if  $r(\vP)$ is inconsistent with $f$. Formally, we have the following definition.
\begin{definition}[\bf Doctrinal paradox]\label{def:DP}
Given any $n\in{\mathbb N}$, any logical connection function $f$, and any quota rule $r$,  a profile $\vP\in\setfrac^n$ is a {\em doctrinal paradox}, if $
r(\vP)(\npm+1) \neq f\big(r(\vP)(1),\cdots,r(\vP)(\npm)\big).$
\end{definition}
The following example illustrates a fractional profile, a quota rule, and the aggregated judgements.

\begin{minipage}[t][][b]{0.3\textwidth}  \begin{example}\label{eg:DP_our}
A (fractional) profile $P$ of three votes, a quota rule $r$ (which uses the majority rule with the breaking in favor of $1$ on both premises and in favor of $\,0$ on the conclusion),  and the aggregation is shown in Table~\ref{tab:eg2}. 

It can be seen from the table that $r(P)=(1,1,0)$, which is inconsistent w.r.t.~$f$. Therefore, $P$ is a doctrinal paradox.
\end{example}
\end{minipage}
\ \ \ \ \ \ 
\begin{minipage}[t][][b]{0.7\textwidth} 
    \begin{tabular}[t]{|@{\ }ccccc@{\ }c@{\ }|}
\hline
      &weight & $\PM_1$ & $\PM_2$  & $\CL$&  $\vv$\\\hline
     \multicolumn{1}{|c|}{Agent 1} & \multicolumn{1}{|c|}{1}  &\multicolumn{1}{|c|}{$1$} & $0$ &  \multicolumn{1}{|c|}{0}  & (0,\,0,\,1,\,0)\\\hline
     \multicolumn{1}{|c|}{Agent 2}   &\multicolumn{1}{|c|}{1} & \multicolumn{1}{|c|}{$0$} & $1$  & \multicolumn{1}{|c|}{0} &  (0,\,1,\,0,\,0)\\\hline
     \multicolumn{1}{|c|}{\multirow{2}*{Agent 3}}   & \multicolumn{1}{|c|}{0.5}  &\multicolumn{1}{|c|}{$1$} & $0$  & \multicolumn{1}{|@{}c@{}|}{0} &  \multirow{2}*{(0,\,0,\,0.5,\,0.5)}\\
     \multicolumn{1}{|c|}{~}   & \multicolumn{1}{|c|}{0.5} &\multicolumn{1}{|c|}{$1$} & $1$ & \multicolumn{1}{|c|}{1} &  ~ \\\hline
     $\HG(\vP)$ & & & & (0,\,1,\,1.5,\,0.5) &\\
     $\nagent_{\ipm}$    && $2$ & $1.5$ &    0.5&   \\\hline
     Breaking $\vbreaking  $   && $1$ & $1$ &    0  &\\
     Threshold $\vthreshold $   && $0.5$& $0.5$ &    0.5 & \\
     $\vthreshold\cdot\nagent$  & & $1.5$ & $1.5$ &   1.5  & \\\hline
    $r_{\vthreshold,\vbreaking}(\vP) $  &  & $1$ & $1$ &    0  &\\
    \hline
    \end{tabular} 
    \captionof{table}{The profile,   rule, and aggregation result for Example~\ref{eg:DP_our}.\label{tab:eg2}}
 \end{minipage}

\section{Characterization of  Smoothed Likelihood of Doctrinal Paradox}
{\bf Intuition.} Before formally present the theorem, we first take a high-level and intuitive approach to introduce the idea behind the characterization and highlight the challenges. Take the max part $\pmaxg$ for example. Suppose the (max-)adversary has chosen  $\vec\pi=(\pi_1,\ldots,\pi_n)\in \Pi^n$. Let $\mathring \pi = \sum_{j=1}^n\pi_j/n$ denote the center of $\vec\pi$. Intuitively,  various multivariate central limit theorems imply that the histogram of the profile $P$ generated from $\vec \pi$ falls in an $\Theta(n^{-0.5})$ neighborhood of $n\mathring\pi$, denoted by ${\cal N}_{n\mathring\pi}$, with high probability.  As an approximation and relation,  for now we allow the max-adversary to choose any $\mathring\pi$ in the convex hull of $\Pi$ (denoted by $\CH(\Pi)$). Then, we can focus on the likelihood of doctrinal paradox in ${\cal N}_{n\mathring\pi}$, which leads to the following four cases.

{\bf Case 1.} If no profile of $n$ votes is a doctrinal paradox, then $\pmaxg=0$ by definition.

{\bf Case 2.}  Otherwise, if for all $\mathring\pi\in \CH(\Pi)$, ${\cal N}_{n\mathring\pi}$ does not contain a doctrinal paradox, then $\pmaxg$ is small. Notice that $\pmaxg$ is not zero, because there is a small (exponential) probability that $\Hist(P)$ is not in ${\cal N}_{n\mathring\pi}$.

{\bf Case 3.}  Otherwise, if the adversary can only choose  $\mathring\pi\in \CH(\Pi)$ that corresponds to a ``non-robust'' instance of doctrinal paradox, in the sense that ${\cal N}_{n\mathring\pi}$ includes some but not too many doctrinal paradoxes. Then, $\pmaxg$ should be larger than that in the previous case. This is the most challenging case, because it is unclear how to characterize conditions for it to happen, and how large the likelihood is. While standard techniques may lead to an $O(n^{-0.5})$ upper bound, it is unclear if $\pmaxg$ can be lower.

{\bf Case 4.}  Otherwise, the adversary can choose $\mathring\pi\in \CH(\Pi)$ that corresponds to a ``robust'' instance of doctrinal paradox, such that ${\cal N}_{n\mathring\pi}$ contains many paradoxes. In this case, $\pmaxg$ is large.  

Case 1 is trivial. It turns out that Case 2 and Case 3 can be characterized by properties of distributions in the convex hull of $\Pi$. 
We first define {\em effective refinements}  of distributions to reason about aggregated judgements in ${\cal N}_{n\mathring\pi}$.

\begin{definition}[\bf (Effective) Refinements]\label{def:breaking}
For any distribution $\pi$ over $\setvote$, any $n\in{\mathbb N}$, and any
quota rule $r$, we say  that a vector $\hvPM = (\hPM_1,\cdots,\hPM_{\npm+1})\in \{0,1\}^{\npm+1}$ is an {\em effective refinement of $\pi$ (at $n$)}, if  the following two conditions hold:\vspace{-0.2em}
\begin{equation}\nonumber
\begin{split}
(1)\;&\forall\,\ipm\in[\npm+1]\;\text{such that}\;\pi\text{ is not tied in }\PM_{\ipm},\;\hPM_{\ipm} = r(\pi)(\ipm), and \\
(2)\;&\exists\,\vP\in\setvote^{\nagent}\;\text{such that}\; r(\vP) = \hvPM.
\end{split}
\end{equation}
If condition (1) holds, then we call $\hvPM$  a {\em refinement} of $\pi$. Let $\calT_{\pi}$ 
denote the set of all effective refinements of $\pi$.
\end{definition}
In words, condition~(1) requires $\hvPM$ to match the aggregated judgements on all non-tied propositions. Condition~(2) requires the existence of a profile of $n$ votes, whose outcome under $r$ is $\hvPM$.
\begin{example}
In the setting of Example~\ref{eg:DP_our},  $\pi = (0.3,0.2,0,0.5)$ is tied in the first premise $(\PM_1)$ and the conclusion. There are four refinements of $\pi$: $(0,1,0)$, $(0,1,1)$ $(1,1,0)$ and $(1,1,1)$.  $(0,1,1)$ is not effective, because no profile can make conclusion $\CL=1$ while keeping $\PM_1$ to be $0$. The other three refinements are effective when $n\geq 3$. For example, $(1,1,1)$ is effective because it is the aggregation of the profile where all agents have $1$ judgements on both premises.
\end{example} 
Recall that $\CH(\Pi)$ denotes the convex hull of $\Pi$. Next, we define four conditions ($\kappa_1$ to $\kappa_4$) to present Theorem~\ref{theo:main}.  While their formal definitions may appear technical, we believe that they have intuitive explanations as discussed right afterward. In fact, $\kappa_1$, $\kappa_2$, and $\kappa_4$ correspond to case 1,2,3 discussed in the beginning of this section, respectively. $\kappa_3$ is defined for min-smoothed likelihood. 

\begin{definition}[\bf \boldmath Conditions $\kappa_1$ to $\kappa_4$] Given $p$ premises, a logical connection function $f$, $n\in\mathbb N$, and  a quota rule $r_{\vthreshold,\vbreaking}$, we define the following conditions.

\noindent {\boldmath $\kappa_1$:} $\forall\,\vP\in\setvote^{\nagent}$, $\vP$ is not a doctrinal paradox. 

\noindent {\boldmath$\kappa_2$:} $\forall\,\pi\in\CH(\Pi)$, $\forall\, \hvPM\in\calT_{\pi}$, $\hvPM$ is consistent (according to $f$). 

 \noindent {\boldmath$\kappa_3$:} $\exists\,\pi\in\CH(\Pi)$ such that $\forall \hvPM\in\calT_{\pi}$, $\hvPM$  is consistent (according to $f$).

\noindent {\boldmath$\kappa_4$:} $\exists\,\ipm\in[m]$ such that ``$\PM_{\npm+1} \leftrightarrow \PM_{\ipm}$ and $\threshold_{\npm+1} = \threshold_{\ipm}$'' or ``$\PM_{\npm+1} \leftrightarrow \neg\,\PM_{\ipm}$ and $\threshold_{\npm+1} = 1-\threshold_{\ipm}$''. 
\end{definition}
Intuitively, $\kappa_1$ requires that no profile of $n$ votes is a doctrinal paradox. $\kappa_2$ requires that the output of $r_{\vthreshold,\vbreaking}$ of every distribution $\pi$ in $\CH(\Pi)$ is ``far away'' from any inconsistent judgements, in the sense that all judgements in $\calT_{\pi}$ are consistent. $\kappa_3$ is weaker than $\kappa_2$, because $\kappa_3$ only requires that the condition in $\kappa_2$   holds for {\em some} $\pi\in\CH(\Pi)$.  
$\kappa_4$ says that the conclusion only relies on one premise and the thresholds are consistent with the relationship between the conclusion and the premise. 

\begin{example}[\bf\boldmath Conditions $\kappa_1$-$\kappa_4$]\label{eg:2}
Let $f$ and $r$ be the same as in  Example~\ref{eg:DP_our}. Let $\Pi = \{\pi_1 = (0.25,0.25,0.25,0.25), \pi_2 = (0.04,0.32,0.32,0.32)\}$. When $n=1$, the paradox does not hold by definition. In the rest of this example, we assume $n\ge 2$.  

\textbf{\boldmath$\kappa_1$   is false.} When $n=2$, $\vP = (1,0,0,1)$ is the only doctrinal paradox. When $n\geq3$, $\vP = (n+1-2\lfloor n/3\rfloor-\lceil n/3\rceil,\lfloor n/3\rfloor,\lfloor n/3\rfloor,\lceil n/3\rceil-1)$ is a doctrinal paradox.\\ 
\textbf{\boldmath$\kappa_2$ is false.} $\CH(\Pi) = \{\pi=a\cdot\pi_1+(1-a)\cdot\pi_2:\,a\in[0,1]\}$. When $a\neq0$, $\pi$ is a doctrinal paradox and contains no ties. When $a=0$, we have $\pi=\pi_1$, and it is tied in conclusion and both premises. Its effective refinement $\hvPM = (1,1,0)$ is a doctrinal paradox.\\ 
\textbf{\boldmath$\kappa_3$ is false.} $\pi_2$ is a doctrinal paradox and contains no ties.\\ 
\textbf{\boldmath$\kappa_4$ is false} according to the definitions of $f$ and $r$.
\end{example}

We say that a distribution $\pi$ is \emph{strictly positive}, if there exists $\epsilon>0$ such that for every $\vPM\in\{0,1\}^{\npm}$, $\pi(\vPM)\geq \epsilon$. A set of distributions $\Pi$ is \emph{strictly positive}, if there exists $\epsilon>0$ such that every $\pi\in\Pi$ is strictly positive (by $\epsilon$).  Being {strictly positive} is a mild requirement for judgement aggregation (c.f.~the argument for voting~\citep{Xia2020:The-Smoothed}). $\Pi$ is \emph{closed} if it is a closed set in $\bbR^{\nvote}$. 
We now present the main theorem of this paper.

\begin{theorem}[\bf Smoothed Likelihood of Doctrinal Paradox]\label{theo:main}
Given any closed and strictly positive set $\Pi$ of distributions over $\setvote$, any logical connection function $f:\setvote\to\{0,1\}$, and any quota rule $r$. For any $\nagent\in\mathbb N$,\vspace{-0.7em}
{
\begin{equation}\nonumber
\begin{split}
&\pmaxg=
\left\{\begin{array}{ll}
0      &       \text{if $\kappa_1$ is true}\\
\exp(-\Theta(\nagent))        & \text{otherwise, if $\kappa_2$ is true} \\
\Theta(\nagent^{-1/2})      &   \text{otherwise, if $\kappa_4$ is true}\\
\Theta(1)      &   \text{otherwise}\\
\end{array}\right.\\
&\pming=
\left\{\begin{array}{ll}
0      &       \text{if $\kappa_1$ is true}\\
\exp(-\Theta(\nagent))        & \text{otherwise, if $\kappa_3$ is true} \\
\Theta(\nagent^{-1/2})      &   \text{otherwise, if $\kappa_4$ is true}\\
\Theta(1)      &   \text{otherwise}\\
\end{array}\right..
\end{split}
\end{equation}}
\end{theorem}

Notice that both max and min smoothed likelihood of doctrinal paradox have four cases: the $0$ case, the exponential case, the polynomial case, and the constant case. The first three cases are good news, particularly in the max part, which states that the doctrinal paradox vanishes when the number of agents is large. The last case is bad news, in particular for the min part, which states that the doctrinal paradox does not vanish.

\noindent We believe that Theorem~\ref{theo:main} is quite general as exemplified in the next subsection, because it works for any $p\ge 1$, any logical connection function, any quote rule, and only makes mild assumptions on $\Pi$. We believe that asymptotic studies as done in Theorem~\ref{theo:main} is important, because first, $n$ is large in some classical applications such as combinatorial voting~\citep{Lang16:Voting}. Second, even though $n$ is typically small in some classical applications such as jury systems, we believe that modern technology will revolutionize such applications in the near future and promote them to the larger scale and broader applications. To be prepared for the future, Theorem~\ref{theo:main} allows us to analyze the practical relevance of the paradox, and provides a mathematical basis for choosing the best (quote) rule to minimize the likelihood of the paradox.

\subsection{Applications of Theorem~\ref{theo:main}}
In this subsection, we present a few examples of applications of  Theorem~\ref{theo:main}. The first example illustrates the $\Theta(1)$ case.
\begin{example}[\bf \boldmath $\Theta(1)$ case]\label{eg:theta1}
Recall that in the setting of  Example~\ref{eg:2}, we have:
\begin{equation}\nonumber
\begin{split}
(\kappa_1,\kappa_2,\kappa_3,\kappa_4) =
\left\{\begin{array}{ll}
(\true,\false,\false,\false)        & \text{if }\nagent=1 \\
(\false,\false,\false,\false)        & \text{if }\nagent\geq2 \\
\end{array}\right.
\end{split}
\end{equation}
\begin{minipage}[t][][b]{0.3\textwidth}
According to Theorem~\ref{theo:main}, the max and min smoothed likelihood of doctrinal paradox are $\Theta(1)$, as shown in Table~\ref{tab:prob2}. A numerical verification can be found in Figure~\ref{fig:1}(a) in Section~\ref{sec:exp}.
\end{minipage}
\hfill
\begin{minipage}[t][][b]{0.67\textwidth}
    {
    \begin{tabular}{cccc}
    \toprule
     Smoothed likelihood & $\nagent=1$  & $\nagent\geq2$ (Ex.~\ref{eg:theta1}) & $\nagent\geq2$ (Ex.~\ref{eg:new}) \\\hline
     $\pmax$   & $0$ & $\Theta(1)$ & $\exp(-\Theta(n))$\\
     $\pmin$   & $0$ & $\Theta(1)$ & $\exp(-\Theta(n))$\\
    \bottomrule
    \end{tabular}}
    \vspace{-0.5em}
     \captionof{table}{Smoothed likelihood of doctrinal paradox  in Example~\ref{eg:theta1} and \ref{eg:new}.\label{tab:prob2}}
     \end{minipage}
\end{example}
The second example illustrates the exponential case, which is positive news.



\begin{example}[\bf Exponentially small case]\label{eg:new}
Let us consider the same quota rule and the same logical connection as in Example~\ref{eg:DP_our} and \ref{eg:2}. Let $\Pi = \{\pi_1 = (0.12, 0.12, 0.12, 0.64), \pi_2 = (0.1,0.1,0.1,0.7)\}$. Following a similar reasoning as in  Example~\ref{eg:2}, we have:
\begin{equation}\nonumber
\begin{split}
(\kappa_1,\kappa_2,\kappa_3,\kappa_4) =
\left\{\begin{array}{ll}
(\true,\false,\false,\false)        & \text{if }\nagent=1 \\
(\false,\true,\true,\false)        & \text{if }\nagent\geq2 \\
\end{array}\right.
\end{split}
\end{equation}
According to Theorem~\ref{theo:main}, the max and min smoothed likelihood of doctrinal paradox are $\exp(-\Theta(n))$, as shown in Table~\ref{tab:prob2}. A numerical verification can be found in Figure~\ref{fig:1}(b).
\end{example}


The following example is also positive news because the doctrinal paradox vanishes as $n\rightarrow\infty$, though not as fast as in Example~\ref{eg:new} w.r.t.~the max-adversary. Notice that the max and min smoothed likelihood of doctrinal paradox are asymptotically different.

\begin{example}[\bf  \boldmath $\Theta(\nagent^{-1/2})$, $\exp(-\Theta(n))$ and $0$ cases]\label{eg:1}
Let the logical connection be $\CL \leftrightarrow \PM_{1}$, $\Pi = \{\pi_1 = (0.9,0.1),\pi_2 = (0.3,0.7)\}$\footnote{Because the doctrinal paradox only depends on the votes to premise $\PM_{\iagent}$ (or conclusion $\CL$), we use the marginal distribution on $\PM_{\ipm}$ to simplify notations. Here, $\pi = (\prob,1-\prob)$ mean $\PM_{\ipm}=0$ with probability $\prob$ while $\PM_{\ipm}=1$ with probability $1-\prob$.} and quota rule $r_{\vthreshold,\vbreaking}$ where $q_1=q_3=0.5$, $d_1 = 1$, and $d_3=0$. 
We have: 
$(\kappa_1,\kappa_2,\kappa_3,\kappa_4) =
\left\{\begin{array}{ll}
(\true,\false,\false,\true)        & \text{if }\nagent\text{ is odd } \\
(\false,\false,\true,\true)        & \text{if }\nagent\text{ is even } \\
\end{array}\right..
$

\begin{minipage}[t][][b]{0.33\textwidth}
According to Theorem~\ref{theo:main}, the max and min smoothed likelihood are presented in  Table~\ref{tab:prob1}. Also see Figure~\ref{fig:1}(c) in Section~\ref{sec:exp} for a numerical verification.
\end{minipage}
\begin{minipage}[t][][b]{0.67\textwidth}
\centering
    \begin{tabular}{ccc}
    \toprule
     Smoothed likelihood & $\nagent$ is odd & $\nagent$ is even \\\hline
     $\pmax$   & $0$ & $\Theta(n^{-1/2})$ \\
     $\DP_{\Pi,r_{\vthreshold,\vbreaking},f}^{\min}(\nagent)$   & $0$ & $\exp(-\Theta(\nagent))$\\
    \bottomrule
    \end{tabular}\vspace{-0.5em}
    \captionof{table}{Smoothed likelihood of doctrinal paradox  in Example~\ref{eg:1}.\label{tab:prob1}}
\end{minipage}
\end{example} 

The following corollary of Theorem~\ref{theo:main} with $\Pi = \{\pi\}$ address the open questions in~\citep{List2005:The-probability} about the probabilities of doctrinal paradox under all i.i.d. distributions, all $q$, all quota rules, and all logical connection funcations.
\begin{corollary}[\bf  \boldmath Likelihood of  doctrinal paradox under i.i.d. distributions as in~\citep{List2005:The-probability}]\label{coro:main}
Given any strictly positive distribution $\pi$ over $\setvote$, any logical connection function $f:\setvote\to\{0,1\}$, and any quota rule $r$. For any $\nagent\in\mathbb N$,
{
\begin{equation}\nonumber
\begin{split}
&{\Pr}_{\vP\sim(\pi)^n}\left(\vP\text{ is a doctrinal paradox}\right)=\left\{\begin{array}{ll}
0      &       \text{if $\kappa_1$ is true}\\
\exp(-\Theta(\nagent))        &  \text{otherwise, if all }  \hvPM\in\calT_{\pi}\text{ are consistent,} \\
\Theta(\nagent^{-1/2})      &   \text{otherwise, if $\kappa_4$ is true}\\
\Theta(1)      &   \text{otherwise}\\
\end{array}\right.,
\end{split}
\end{equation}}
\end{corollary}
In particular, when $\pi$ is the uniform distribution over $\setvote$ and the aggregation rule is the majority, the likelihood of doctrinal paradox  is either $\Theta(1)$ or $0$ depending on the logical connection function $f$. 

\subsection{Proof Sketch of Theorem~\ref{theo:main}}\label{sec:tech_proof}
The proof proceeds in three steps. In {\bf Step 1}, we model  doctrinal paradox  as  unions of  {polyhedra}. In {\bf Step 2}, we apply~\citep[Theorem 2]{Xia2021:How-Likely} to obtain a characterization. In {\bf Step 3}, we close the gap between  Step 2 and Theorem~\ref{theo:main} by translating conditions and degree of polynomial in Step 2 to their counterparts in Theorem~\ref{theo:main}. 


{\noindent \bf Step 1: Polyhedra Presentation. } We show that the region of doctrinal paradox can be presented by a set of polyhedra, each of which is in the $\nvote$-dimensional Euclidean space with the form $\calH = \left\{\vx:\vA\cdot\vx\leq\vb\right\}$. For any $\ipm\in[\npm+1]$, we first define vector $\vc_{\ipm}\in\{\threshold_\ipm-1,\threshold_\ipm\}^\nvote$ as follows:
\begin{equation}\label{equ:vc}
\vc_\ipm(\vPM) \triangleq \left\{\begin{array}{ll}
\threshold_\ipm-1    &       \text{if } \vPM\in\setidx_{\ipm}\\
\threshold_\ipm      &       \text{otherwise}
\end{array}\right..
\end{equation}
For any profile $P$, it is not hard to verify that the sign of $\vc_{\ipm}\cdot\HG(\vP)$ is closely related to the aggregation result by the quota rule. We define sign function $\sign(\cdot)$ and $\bbone(\cdot)$ as follows,
\begin{equation}\nonumber
\begin{split}
\sign(x)&\triangleq\left\{\begin{array}{lcl}
1      &      & \text{if } x > 0\\
-1      &      & \text{otherwise}
\end{array}\right.
\;\;\;\text{and}\;\;\;\;\;\bbone(\kappa)\triangleq\left\{\begin{array}{lcl}
1      &      & \text{if } \kappa \text{ is true}\\
0\;\;\;      &      & \text{otherwise}
\end{array}\right..\\
\end{split}
\end{equation} 
Then, we define
$
\vA_{\hvPM}  \triangleq \begin{pmatrix}
\sign(\hPM_1)\cdot\vc_1^{\T}\\
\vdots\\
\sign(\hPM_{\npm+1})\cdot\vc_{\npm+1}^{\T}
\end{pmatrix}$ and $\vb_{\hvPM} \triangleq \begin{pmatrix}
\sign(\hPM_1)\cdot \breaking_1-\bbone(\hPM_1=1)\\
\vdots\\
\sign(\hPM_{\npm+1})\cdot \breaking_{\npm+1}-\bbone(\hPM_{\npm+1}=1)
\end{pmatrix}.
$
Let $\calH_{\hvPM} = \{\vx:\;\vA_{\hvPM}\,\vx\leq\vb_{\hvPM}\}$ denote the polyhedra and let $\calH_{\hvPM,\;\leq0} = \{\vx:\;\vA_{\hvPM}\,\vx\leq\vzero\}$ denote its \emph{characteristic cone}.  Let $\calC = \{\calH_{\hvPM}:\;\hvPM\in\{0,1\}^{\npm+1}\text{ is inconsistent}\}$
denote the set of all polyhedra corresponding to doctrinal paradox.  The following example visualizes a polyhedron and its characteristic cone for doctrinal paradox. 

\begin{example}[\bf Polyhedra representation of doctrinal paradox]\label{eg:polyhedra}
We consider a system with one premise $\PM_1$. The logical connection between conclusion $\CL$ and premise $\PM_1$ is $\CL\leftrightarrow\PM_1$. The parameters for quota rule are $(\threshold_1,\threshold_2)=(0.25,0.65)$ and $(\breaking_1,\breaking_2)=(1,1)$. 

\begin{minipage}[t][][b]{0.5\textwidth}
According to Equation~(\ref{equ:vc}), we have $\vc_1 = (\threshold_1,\threshold_1-1)^{\T}$ and $\vc_2 = (\threshold_2,\threshold_2-1)^{\T}$. 
Then, the aggregation results of $\hvPM=(1,0)$ and $\hvPM=(0,1)$ are both inconsistent. Thus, $\calC = \left\{\calH_{(1,0)},\calH_{(0,1)}\right\}$.  $\calH_{(1,0)}$ is represented by $\vA_{(1,0)}$ and $\vb_{(1,0)}$ defined as follows.

$
\vA_{(1,0)}  = \begin{pmatrix}
\threshold_1 & \threshold_1-1\\
-\threshold_2 & 1-\threshold_2
\end{pmatrix} \text{ and }\vb_{(1,0)} = \begin{pmatrix}
0\\
-1
\end{pmatrix}$

It is not hard to verify that for any profile $\vP$ of $\nagent$ agents, $\vA_{(1,0)}\cdot\HG(\vP)\leq\vb_{(1,0)}$ if and only if
\begin{equation*}
\nagent_1\geq\threshold_1\cdot\nagent\;\;\;\text{and}\;\;\;
\nagent_1\leq\threshold_2\cdot\nagent-1
\end{equation*}
where $n_1$ represents the number of votes for $\PM_1=1$ in $\vP$. Figure~\ref{fig:eg_polyhedra} illustrates the regions corresponding to   $\calH_{(1,0)}$ and its characteristic cone. 
\end{minipage}
\hfill
\begin{minipage}[t][][b]{0.46\textwidth}
\includegraphics[width =  \textwidth]{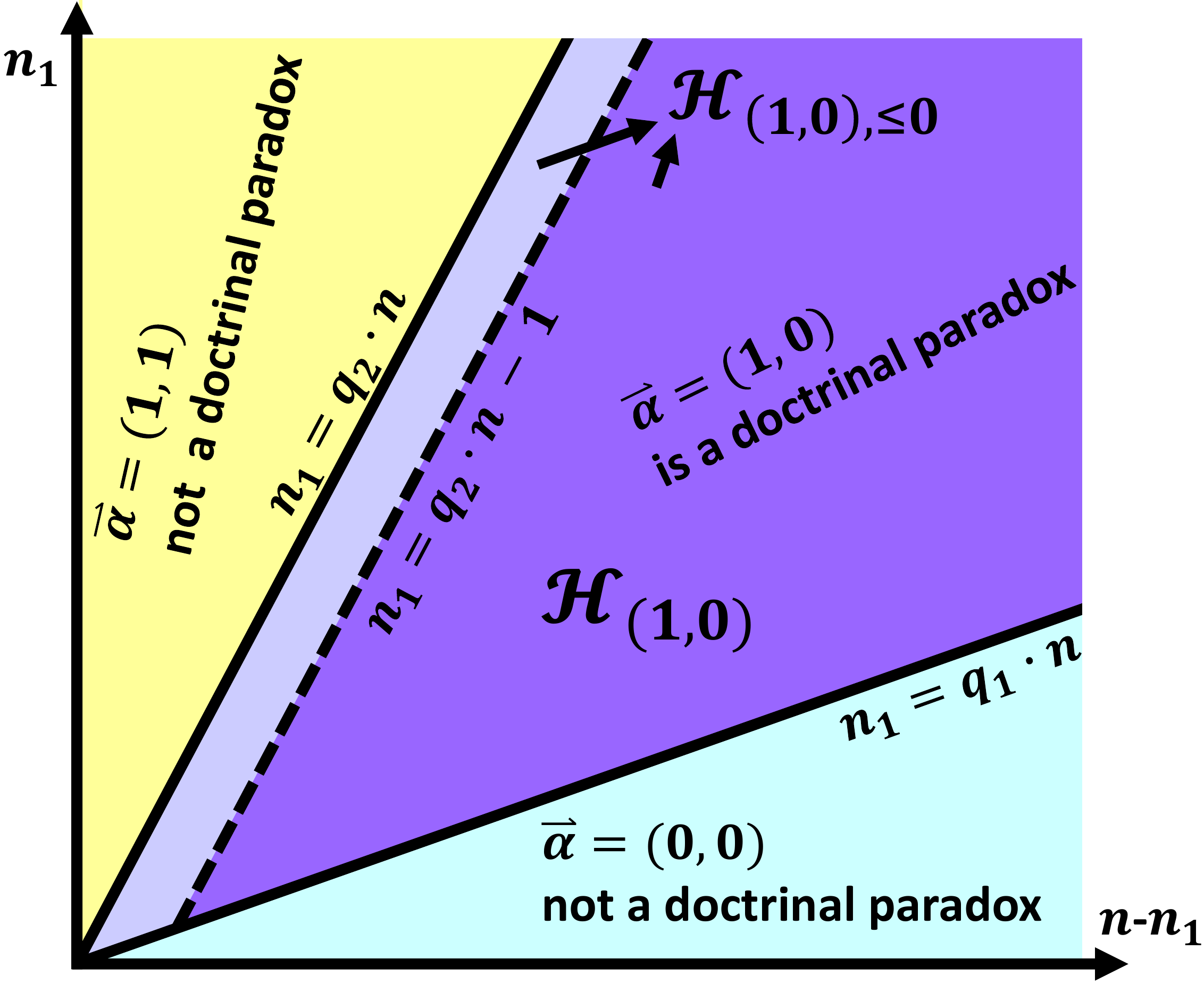}
\captionof{figure}{\small  $\calH_{(1,0)}$ in Example~\ref{eg:polyhedra}. The characteristic cone $\calH_{(1,0),\,\leq0}$ is the combination of $\calH_{(1,0)}$ and the region between $\nagent_1 = \threshold_2\cdot\nagent-1$ and $\nagent_1 = \threshold_2\cdot\nagent$.}
\label{fig:eg_polyhedra}
\end{minipage}

\end{example}

Consequently, the smoothed likelihood of doctrinal paradox is equivalent to the smoothed probability for the histogram of the randomly generated profile (which is a {\em Poisson multivariate variable (PMV)}) to be in $\calC$. Therefore, in Step 2 we apply~\citep[Theorem 2]{Xia2021:How-Likely} to obtain a characterization for the PMV-in-$\calC$ problem. However, the results are too technical and generic to be informative in judgement aggregation context (see Lemma~\ref{theo:tech} in Appendix~\ref{sec:main_proof}). Step 3 aims at obtaining the characterization as stated in Theorem~\ref{theo:main} based on Lemma~\ref{theo:tech}, which involves non-trivial calculations and simplifications. The full proof can be found in Appendix~\ref{sec:main_proof}. $\hfill\Box$

\subsection{Computational Verifications of Conditions in Theorem~\ref{theo:main}}
In practice, $\kappa_1$-$\kappa_4$ may not be easy to verify by hand. In this subsection, we briefly discuss algorithms for checking $\kappa_1$-$\kappa_4$  under any quota rule $r_{\vthreshold,\vbreaking}$ with rational threshold $\vthreshold$ w.r.t.~any $\Pi$ that consists of finitely many distributions whose probabilities are rational numbers. The runtime of these algorithms is polynomial in $m$ (as well as $\log n$ for $\kappa_1$). While $m=2^p$ may seem exponentially large, notice that the input size (e.g., $\Pi$ and the truth-table representation of $f$) can also be large, i.e., $\Theta(m)$.

Suppose $ \Pi = \{\pi^1,\ldots,\pi^\ell\}\subseteq {\mathbb Q}^{m}$. $\kappa_1$-$\kappa_4$ can be verified as follow, by combining the results of multiple (I)LP whose variables are $\vec y = (y^1,\ldots,y^\ell)$ that represent  the multiplicity of $\pi^\ell$. 

{\bf \boldmath $\kappa_1$} holds for any given $n$ if and only if there exists an inconsistent $\vec\alpha\in \{0,1\}^{\npm+1}$ such that the following ILP is feasible.
\begin{equation}
\label{eq:ilp}
\begin{pmatrix*}[r]
\vA_{\hvPM} \cdot ((\pi^1)^{\T},&\ldots&, (\pi^\ell)^{\T})\\
&\vec 1&\\
&-\vec 1&\\
\end{pmatrix*}\cdot (\vec y) \le \begin{pmatrix*}[r]
\vb_{\hvPM}\\
n\\
-n\\
\end{pmatrix*}, \text{ where }\vec y \text{ are integers}.
\end{equation}
Following Lenstra's theorem~\cite{Lenstra83:Integer}, when $\ell$ is viewed as a constant (and $q$ and $n$ are variables), the runtime of each ILP is polynomial in $q$ and $\log n$. Therefore, the overall complexity of checking $\kappa_1$ is polynomial in $\nvote$ ($=2^{\npm}$) and $\log n$. Let $A_n\subseteq \{0,1\}^{p+1}$ denote the set of all inconsistent $\alpha$ such that (\ref{eq:ilp}) has a feasible (non-negative integer) solution, which will be used later.

{\bf \boldmath $\kappa_2$ and $\kappa_3$} can be verified by similar algorithms. Notice that $\calT_{\pi}$ only depends on whether each component of $r(\pi)$ is positive, $0$, or negative. Therefore, we can naturally define $\calT_{\vec\beta}$ for every $\vec\beta\in \{+,0,-\}^{\npm+1}$. Let $B_n$ denote the set of all  $\vec \beta$, for each of which there exists $\alpha\in A_n$ that refines it. 
Then, for every $\vec\beta\in B_n$, we define LP$_{\vec\beta}$, whose variables are $\vec y\ge \vec 0$ that do  not need to be integers, to verify whether there exists $\pi\in \CH(\Pi)$ such that $r_{\vthreshold,\vbreaking}(\pi) = \vec \beta$. More precisely,  for every $j\le p+1$, \\
$1^{\circ}$ if $\beta_j = +$, then LP$_{\vec\beta}$ contains a linear constraint $\vc_j^{\T}\cdot (\sum_{j=1}^\ell y_j(\vec \pi^j)^{\T}) \le -1$;\\
$2^{\circ}$ if $\beta_j = -$, then  LP$_{\vec\beta}$ contains a linear constraint $-\vc_j^{\T}\cdot (\sum_{j=1}^\ell y_j(\vec \pi^j)^{\T}) \le -1$;\\
$3^{\circ}$ if $\beta_j = 0$, then  LP$_{\vec\beta}$ contains two linear constraints $ \vc_j^{\T}\cdot (\sum_{j=1}^\ell y_j(\vec \pi^j)^{\T}) \le 0$ and $-\vc_j^{\T}\cdot (\sum_{j=1}^\ell y_j(\vec \pi^j)^{\T}) \le 0$.\\
Therefore, $\kappa_2$ is true  if for every $\vec\beta\in B_n$, LP$_{\vec\beta}$ is not feasible. $\kappa_3$ is true, if there exists $\vec\beta\notin B_n$ such that LP$_{\vec\beta}$ is feasible. Notice that $\ell$ is viewed as a constant, which means that LP$_{\vec\beta}$ can be solved in time that is polynomial in $m$. 

{\bf \boldmath $\kappa_4$} can be verified by enumerating all inputs of $f$, which takes time that is polynomial in $m$.

\section{Experiments}\label{sec:exp}
We conduct numerical experiments to verify the results in Theorem~\ref{theo:main}. The first three experiments (Figure~\ref{fig:1}) follows the same setting as Example~\ref{eg:theta1},~\ref{eg:new},~\ref{eg:1} respectively.
In Figure~\ref{fig:1}, $\text{DP}_{\Pi}^{\max}$ and $\text{DP}_{\Pi}^{\min}$ (the blue circles and red stars) represent the estimated max-smoothed likelihood and the min-smoothed likelihood of doctrinal paradoxes. The dot curves illustrate the fittings of the estimated smoothed probabilities. The expressions and fitness of all fitting curves are presented in Appendix~\ref{sec:add_exp}. Recalling the notations used in our definition of $\pmaxg$, we say $\vpi\in\Pi^{\nagent}$ is one kind of \emph{distribution assignment}. 
We run one million ($10^6$) independent trials to estimate the probability of doctrinal paradoxes under each distribution assignment. Then, $\text{DP}_{\Pi}^{\max}$ (or $\text{DP}_{\Pi}^{\min}$) takes the maximum (or the minimum) probability of doctrinal paradoxes among all distribution assignments. 

\vspace{-0.5em}
\begin{figure*}[ht]
\includegraphics[width = \textwidth]{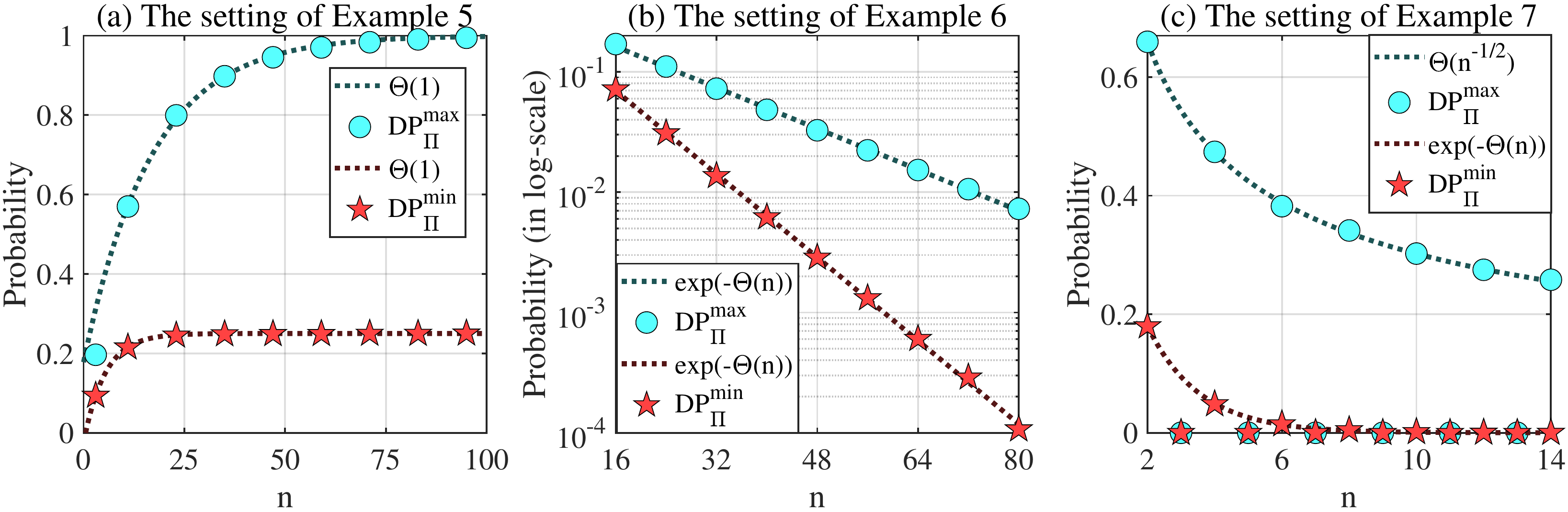}
\vspace{-1.7em}
\caption{\small  A numerical verification of Theorem~\ref{theo:main}. Note that the plot for Example~\ref{eg:new} is in log-scale.}
\label{fig:1}
\vspace{-0.5em}
\end{figure*}

It is easy to see that the results are consistent with  Theorem~\ref{theo:main}. For example, Figure~\ref{fig:1}(a) shows that both the max-smoothed likelihood and the min-smoothed likelihood of doctrinal paradoxes are $\Theta(1)$, which matches the result in Table~\ref{tab:prob2}. Two additional sets of experiments, whose results are presented in Figure~\ref{fig:2}(b) in Appendix~\ref{sec:add_exp}, study a more complex setting of three premises.

\section{Conclusions and Future Work}\label{sec:con}
We made a first step towards understanding the likelihood of doctrinal paradox in the smoothed analysis framework. There are many immediate open questions. For example, can we characterize the smoothed likelihood of the paradox for general, more complicated propositions (sometimes called {\em agendas})? Can we remove the strict positiveness assumption on $\Pi$? Can we prove smoothed impossibility theorems that extend the quantitative impossibility theorems under i.i.d.~uniform distributions~\citep{Nehama2013:Approximately,Filmus2020:AND-testing}? More generally, we believe that building a comprehensive picture of the smoothed properties of aggregation rules in judgement aggregation is a promising (yet challenging) direction in theory and in practice.

{
\bibliographystyle{plainnat} 

}

\clearpage
\appendix
\onecolumn
{\large \centerline{\textbf{Supplementary Material for The Smoothed Likelihood of Doctrinal Paradox}}
\centerline{\textbf{}}}

\section{Detailed Settings for Section~\ref{sec:prelim} (Preliminary)}
\subsection{Vectorized indexing for votes or histogram}\label{sec:extra_notations}
We fix the judgement of $\vPM = \vell = (\ell_1,\cdots,\ell_{\npm})$ to correspond to the $(\overline{\ell_1\cdots\ell_{\npm}}_{(2)}+1)$-th component of $\vv$, where $\overline{\ell_1\cdots\ell_{\npm}}_{(2)}$ is the binary number taking $\ell_\ipm$ as the $\ipm$-th digit. For example, if $\PM_{\iagent,1}=1$ and $\PM_{\iagent,2}=0$, we will have $\vv_{\iagent}(\overline{10}_{(2)}+1) = \vv_{\iagent}(3) =1$ while all other components of $\vv_{\iagent}$ are zero. 


\subsection{Properties of Quota Rules}\label{sec:extra_quota}
The properties of quota rules include :\\
\emph{Anonymity}: each pair of agents play interchangeable roles.\\
\emph{Neutrality}: each pair of propositions are interchangeable.\\
\emph{Independence}: the aggregation of a certain proposition only depends on agents' judgement on this proposition.\\
\emph{Monotonicity}: adding vote to the current winner will never change the winner

\section{Proof of Theorem~\ref{theo:main}}\label{sec:main_proof}
\textbf{Theorem~\ref{theo:main} (Smoothed Likelihood for Doctrinal Paradox)\textbf{. }}  \emph{
Given any strictly positive set $\Pi$ of distributions over $\setvote$, any logic function $f:\setvote\to\{0,1\}$, and any quota rule $r$. For any $\nagent\in\bbZ_{>0}$,
{
\begin{equation}\nonumber
\begin{split}
&\pmaxg=
\left\{\begin{array}{ll}
0      &       \text{if $\kappa_1$ is true}\\
\exp(-\Theta(\nagent))        & \text{otherwise, if $\kappa_2$ is true} \\
\Theta(\nagent^{-1/2})      &   \text{otherwise, if $\kappa_4$ is true}\\
\Theta(1)      &   \text{otherwise}\\
\end{array}\right.\\
&\pming=
\left\{\begin{array}{ll}
0      &       \text{if $\kappa_1$ is true}\\
\exp(-\Theta(\nagent))        & \text{otherwise, if $\kappa_3$ is true} \\
\Theta(\nagent^{-1/2})      &   \text{otherwise, if $\kappa_4$ is true}\\
\Theta(1)      &   \text{otherwise}\\
\end{array}\right..
\end{split}
\end{equation}}
}
\begin{proof} We present the full proof for Step 3 discussed in the sketch in the main text.

{\bf Step 1: Polyhedra Representation. } We prove the following lemma on the connection between $\calC$ and the doctrinal paradox.
\begin{lemma}[\bf Polyhedra representation of doctrinal paradox]\label{lem:polyhedra}
Given any profile $\vP\in\setvote^{\nagent}$ and any distribution $\pi$ on $\setvote$, any quota rule $r$ and any logical connection function 
$f$, we have the following two statements for doctrinal paradox,
\begin{equation}\nonumber
\begin{split}
(1)\;\;\;&\forall\,\calH\in\calC,\;\HG(\vP)\notin\calH
\Leftrightarrow\;\vP \text{ is not a doctrinal paradox}.\\
(2)\;\;\;&\forall\,\calH\in\calC,\;\pi\notin\calH_{\leq0}
\Leftrightarrow\;\text{all refinements of $\pi$ are consistent}.\\
\end{split}
\end{equation}
\end{lemma}
\begin{proof}
According to the definition of $\calC$, we only need to prove the following statement for polyhedra:
\begin{statement}\label{equ:statement}
For any profile $\vP\in\setvote^{\nagent}$ and any distribution $\pi$ on $\setvote$, any quota rule $r_{\vthreshold,\vbreaking}$ and any logical connection function $f$,
\begin{equation}\nonumber
\begin{split}
\HG(\vP)\in\calH_{\hvPM}\;\;&\Leftrightarrow\;\;r_{\vthreshold,\vbreaking}(\vP) = \hvPM\;\;\;\;\;\;\;\;\;\;\text{and}\\
\pi\in\calH_{\hvPM,\,\leq0}\;\;\;\;&\Leftrightarrow\;\;\hvPM\text{ is a refinement of }\pi.
\end{split}
\end{equation}
\end{statement}

Let $\vPM_{\vell}$ be the corresponding type of vote when the vote on premises is $\vell$. Mathematically, $\vPM_{\vell}\in\setvote$ and $\vPM_{\vell}(\vell)=1$. 
For all $\ipm\in[\npm+1]$, the characterization vector $\vc_{\ipm}$ can be written as,
$\vc_{\ipm} = (\threshold_\ipm-1)\cdot\left(\sum_{\vell\in\setidx_{\ipm}}\vPM_{\vell}\right) + \threshold_\ipm\left(\sum_{\vell\notin\setidx_{\ipm}}\vPM_{\vell}\right)$. Then,
{
\begin{equation}\nonumber
\begin{split}
\vc_{\ipm}^{\T}\cdot\HG(\vP)
=\;& \vc_{\ipm}^{\T}\cdot\left(\sum_{\vell}\HG(\vP)(\vell)\cdot\vPM_{\vell}\right)\\
=\;& (\threshold_\ipm-1)\cdot\left(\sum_{\vell\in\setidx_{\ipm}}\HG(\vP)(\vell)\right) + \threshold_\ipm\left(\sum_{\vell\notin\setidx_{\ipm}}\HG(\vP)(\vell)\right)\\
=\;& n\cdot \threshold_\ipm - \underbrace{\sum_{\vell\in\setidx_{\ipm}}\HG(\vP)(\vell)}_{\text{number of votes for $\PM_\ipm=1$}}.
\end{split}
\end{equation}}
According to the definition of quota rule, we know that,
\begin{equation}\nonumber
\begin{split}
\vc_{\ipm}^{\T}\cdot \HG(\vP) \leq \breaking_\ipm-1 \;\;&\Leftrightarrow\;\;r_{\vthreshold,\vbreaking}(\vP)(\ipm)= 1\;\;\text{and}\\
-\vc_\ipm^{\T}\cdot \HG(\vP) \leq -\breaking_\ipm \;\;&\Leftrightarrow\;\;r_{\vthreshold,\vbreaking}(\vP)(\ipm)= 0\\
\end{split}
\end{equation}
Then, the first part of Statement~\ref{equ:statement} follows by the definition of $\calH_{\hvPM}$. Follow similar procedure as above, we have,
\begin{equation}\nonumber
\begin{split}
\vc_{\ipm}^{\T}\cdot\pi &= \threshold_\ipm - \sum_{\vell\in\setidx_{\ipm}}\pi(\vell)
\end{split}
\end{equation}
Thus, for any distribution $\pi$ over $\setvote$, 
\begin{equation}\nonumber
\begin{split}
\vc_{\ipm}^{\T}\cdot \pi \leq 0 \;\;&\Leftrightarrow\;\;\exists \,\hvPM\in\calT_{\pi} \text{ such that } \hvPM(\ipm)= 1\;\;\text{and}\\
-\vc_{\ipm}^{\T}\cdot \pi \leq 0 \;\;&\Leftrightarrow\;\;\exists \,\hvPM\in\calT_{\pi} \text{ such that } \hvPM(\ipm)= 0\\
\end{split}
\end{equation}
Then, the second part of Statement~\ref{equ:statement} follows by the definition of $\calH_{\hvPM,\,\leq0}$.
\end{proof}

{\bf Step 2: Apply~\citep[Theorem 2]{Xia2021:How-Likely}. } Following~\cite{Xia2021:How-Likely}, for any distribution $\pi$ over $\setvote$ and any $\nagent\in\mathbb N$, the {\em active dimension} of $\calH$ is formally defined as:
$$
\dim_{\calH,\nagent}(\pi)\triangleq
\left\{\begin{array}{ll}
\dim(\calH_{\leq 0})      &      \text{if }\pi\in\calH_{\leq0} \text{ and } \exists\,\vP\in\setvote^{\nagent}\text{ such that } \HG(\vP)\in\calH \\
-\infty      &   \text{otherwise }
\end{array}\right..
$$
We say a polyhedra is active (to $\pi$ at $n$) if and only if $\dim_{\calH,\nagent}(\pi)\neq-\infty$. In words, an active polyhedron requires (1) there exists a profile of $n$ votes, whose histogram is in the polyhedron, and (2) $\pi$ is in the characteristic cone of the polyhedron. In the next lemma, we reveal a relationship between the active dimensions of polyhedra in $\calC$ and the smoothed probability of the doctrinal paradox.

\begin{lemma}[Active Dimension$\,\to\,$Smoothed Probability, a direct application of Theorem 2 in~\citep{Xia2021:How-Likely}]
\label{theo:tech}
Given any closed and strictly positive set $\Pi$ of distributions over $\setvote$, any logical connection function $f:\setvote\to\{0,1\}$ and any quota rule $r$. For any $\nagent\in\mathbb N$,
{
\begin{equation}\nonumber
\begin{split}
&\pmaxg= \left\{\begin{array}{ll}
0     &   \text{if C1 is true}\\
\exp(-\Theta(\nagent))& \text{otherwise, if C2 is true} \\
\max_{\calH\in\calC}\Theta\left(n^{(\dim_{\calH}(\pi)-m)/2}\right) & \text{otherwise}\\
\end{array}\right.\\
&\pming= \left\{\begin{array}{ll}
0      & \text{if C1 is true}\\
\exp(-\Theta(\nagent)) &\text{otherwise, if C3 is true} \\
\min_{\calH\in\calC}\Theta\left(n^{(\dim_{\calH}(\pi)-m)/2}\right)& \text{otherwise}\\
\end{array}\right.,\\
\end{split}
\end{equation}}
\noindent where C1 is: $\forall\vP\in\setvote^{\nagent}$, $\vP$ is not doctrinal paradox.\\ 
C2 is: $\forall\, \pi\in\CH(\Pi)\text{ and }\forall\,\calH\in\calC,\;\dim_{\calH}(\pi) = -\infty$ \\ 
C3 is: $\exists\, \pi\in\CH(\Pi)\text{ and }\forall\,\calH\in\calC,\;\dim_{\calH}(\pi) = -\infty$.
\end{lemma}
\noindent Note that C1 is the same as $\kappa_1$ in Theorem~\ref{theo:main}. By Lemma~\ref{lem:polyhedra}, we know that $\kappa_2$ (or $\kappa_3$) is very similar with C2 (or C3), which says that all (or some) distributions in $\CH(\Pi)$ are ``far'' from all active polyhedra of doctrinal paradoxes.  

{\noindent \bf Step 3: Closing the Gap  between Lemma~\ref{theo:tech} and Theorem~\ref{theo:main}.} The only gap left between Lemma~\ref{theo:tech} and Theorem~\ref{theo:main} is that the active dimension $\dim_{\calH,\nagent}(\pi)$ is still unknown, which will be characterized in Step 3.1 and 3.2 as follows.


{\vspace{2mm}\noindent \bf Step 3.1: determine $\dim(\calH_{\leq 0})$ (the dimensionality of characteristic cones). }

{\vspace{1mm}\noindent \bf Step 3.1.1: the linear independence of $\{\vc_1,\cdots,\vc_{\npm}\}$. }
We first show that $\vc_1,\cdots,\vc_{\npm}$ are linearly independent by contradiction. For any $\ipm\in[\npm]$, we assume that $\vc_{\ipm} = \sum_{\ipm'\neq\ipm}a_{\ipm'}\cdot\vc_{\ipm'}$, where  $a_{\ipm'}$ are real valued coefficients. Because $\vc_{\ipm}$ is not a zero vector and different with any $\vc_{\ipm'}$, there must exist a pair of $\ipm_1',\ipm_2'\in[\npm]\setminus\{\ipm\}$ such that $a_{\ipm_1'}\neq0$ and $a_{\ipm_2'}\neq0$. Let $\vc_{\ipm}{(\xx_1,\xx_2,\vzero)}$ to be the coordinate of $\vc_{\ipm}$ corresponds to $\PM_{\ipm_1'}=\xx_1$, $\PM_{\ipm_2'}=\xx_2$ and all other premises are $0$. Because $\vc_{\ipm}$ has the same value in all components that $\PM_{\ipm}=0$. We should have $\vc_{\ipm}{(0,0,\vzero)} = \vc_{\ipm}{(1,0,\vzero)} = \vc_{\ipm}{(1,1,\vzero)}$, which is contradict with the observation that
{
\begin{equation}\nonumber
\begin{split}
\vc_{\ipm}{(0,0,\vzero)} &= a_{\ipm_1'}\cdot \threshold_{\ipm_1'} +  a_{\ipm_2'}\cdot \threshold_{\ipm_2'} +\sum_{\ipm'\notin\{\ipm_1'.\ipm_2'\}}a_{\ipm'}\cdot \threshold_{\ipm'}\\
\vc_{\ipm}{(1,0,\vzero)} &= a_{\ipm_1'} (\threshold_{\ipm_1'}-1) +  a_{\ipm_2'}\cdot \threshold_{\ipm_2'} +\sum_{\ipm'\notin\{\ipm_1'.\ipm_2'\}}a_{\ipm'}\cdot \threshold_{\ipm'}\\
\vc_{\ipm}{(1,1,\vzero)} &= a_{\ipm_1'} (\threshold_{\ipm_1'}-1) +  a_{\ipm_2'} (\threshold_{\ipm_2'}-1) +\sum_{\ipm'\notin\{\ipm_1'.\ipm_2'\}}a_{\ipm'}\cdot \threshold_{\ipm'}\;\;\;\;\;\;\text{ and thus}\\
\vc_{\ipm}{(0,0,\vzero)}&-\vc_{\ipm}{(1,0,\vzero)} = a_{\ipm_1'}\neq 0\\
\vc_{\ipm}{(1,0,\vzero)}&-\vc_{\ipm}{(1,1,\vzero)} = a_{\ipm_2'}\neq 0.
\end{split}
\end{equation}}

{\vspace{1mm}\noindent \bf Step 3.1.2: the dimension of $\calH_{\hvPM}$ when $\kappa_4$ is true. } 
When $\kappa_4$ is true, the logical connection can be denoted as either $\CL\leftrightarrow\PM_{\ipm}$ or $\CL\leftrightarrow\neg\,\PM_{\ipm}$. For both cases, we have $\sign_{0\to -1}(\hPM_{\npm+1})\cdot\vc_{\npm+1} = -\sign_{0\to -1}(\hPM_{\ipm})\cdot\vc_{\ipm}$. When $\kappa_4$ is true, we will have , the following equation must hold in $\calH_{\hvPM,\,\leq0}$.
\begin{equation}\nonumber
\begin{split}
&\vx^{\T}\cdot\left(\sign_{0\to -1}(\hPM_{\npm+1})\cdot\vc_{\npm+1}\right)\\
=\;&\vx^{\T}\cdot\left(\sign_{0\to -1}(\hPM_{\ipm})\cdot\vc_{\ipm}\right) = 0.
\end{split}
\end{equation}
Thus, when $\kappa_4$ is not true, we draw the conclusion that
\begin{equation}\nonumber
\forall\,\calH\in\calC,\;\;\dim(\calH_{\leq0}) = \nvote-1 \;\;\text{ when  $\kappa_4$ is true}.
\end{equation}


{\vspace{1mm}\noindent \bf Step 3.1.3: the dimension of $\calH_{\hvPM}$ when $\kappa_4$ is not true. }
We first show the linear independence between $\vc_{\npm+1}$ and $\{\vc_1,\cdots,\vc_{\npm}\}$ when $\kappa_4$ is false. We assume that $\vc_{\npm+1} = \sum_{\ipm'\in[\npm]} a_{\ipm'}\cdot\vc_{\ipm'}$, where $a_{\ipm'}$ are real valued coefficients. Because $\vc_{\npm+1}$ is not a zero vector and different with any $\vc_{\ipm'}$ when $\kappa_4$ is not true, there must exist a pair of $\ipm_1',\ipm_2'\in[\npm]\setminus\{\ipm\}$ such that $a_{\ipm_1'}\neq0$ and $a_{\ipm_2'}\neq0$. 
Let $\vc_{\npm+1}{(\xx_1,\xx_2,\vzero)}$ to be the coordinate of $\vc_{\npm+1}$ corresponds to $\PM_{\ipm_1'}=\xx_1$, $\PM_{\ipm_2'}=\xx_2$ and all other premises are $0$. Thus, we have the following observations: 
{
\begin{equation}\label{equ:obv}
\begin{split}
\vc_{\npm+1}{(0,0,\vzero)} &= a_{\ipm_1'}\cdot \threshold_{\ipm_1'} +  a_{\ipm_2'}\cdot \threshold_{\ipm_2'} +\sum_{\ipm'\notin\{\ipm_1'.\ipm_2'\}}a_{\ipm'}\cdot \threshold_{\ipm'}\\
\vc_{\npm+1}{(1,0,\vzero)} &= a_{\ipm_1'} (\threshold_{\ipm_1'}-1) +  a_{\ipm_2'}\cdot \threshold_{\ipm_2'} +\sum_{\ipm'\notin\{\ipm_1'.\ipm_2'\}}a_{\ipm'}\cdot \threshold_{\ipm'}\\
\vc_{\npm+1}{(1,1,\vzero)} &= a_{\ipm_1'} (\threshold_{\ipm_1'}-1) +  a_{\ipm_2'} (\threshold_{\ipm_2'}-1) +\sum_{\ipm'\notin\{\ipm_1'.\ipm_2'\}}a_{\ipm'}\cdot \threshold_{\ipm'}\;\;\;\;\;\;\text{ and thus}\\
\vc_{\npm+1}{(0,0,\vzero)}&-\vc_{\npm+1}{(1,0,\vzero)} = a_{\ipm_1'}\neq 0\\
\vc_{\npm+1}{(1,0,\vzero)}&-\vc_{\npm+1}{(1,1,\vzero)} = a_{\ipm_2'}\neq 0.
\end{split}
\end{equation}}

{\noindent \emph{Case 3.1.3.1}}:  $\vc_{\npm+1}{(0,0,\vzero)}-\vc_{\npm+1}{(1,1,\vzero)} = a_{\ipm_1'}+a_{\ipm_2'} \neq 0$. For this case, there must exist at least three different values in $\vc_{\ipm}$, which contradict with the fact that $\vc_{\npm+1}\in\{\threshold_{\npm+1}-1,\threshold_{\npm+1}\}^{\nvote}$. 

{\noindent \emph{Case 3.1.3.2}}:  $\vc_{\npm+1}{(0,0,\vzero)}-\vc_{\npm+1}{(1,1,\vzero)} = a_{\ipm_1'}+a_{\ipm_2'}  = 0$. 
we note that any nonzero value of $a_{\ipm_1'}+a_{\ipm_2'}$ already results in contradiction according to \emph{Case 1.1.1}. Then, the above relationship of $a_{\ipm_1'}+a_{\ipm_2'}=0$ must holds for all nonzero $a_{\ipm'}$s, which implies that there cannot exists more than three non-zero values for $a_{\ipm'}$.  
Now, the only case left is $\vc_{\ipm} = a(\vc_{\ipm'_1}-\vc_{\ipm'_2})$, where $a$ is a nonzero real number. 
According to the observations in Inequalities~\ref{equ:obv}, we know that 
{
\begin{equation}\nonumber
\begin{split}
\vc_{\npm+1}{(0,0,\vzero)} &= a(\threshold_{\ipm_1'} - \threshold_{\ipm_2'})\\
\vc_{\npm+1}{(1,0,\vzero)} &= a(\threshold_{\ipm_1'} - \threshold_{\ipm_2'}-1)\\
\vc_{\npm+1}{(0,1,\vzero)} &= a(\threshold_{\ipm_1'} - \threshold_{\ipm_2'}+1),\\
\end{split}
\end{equation}}
which says that $\vc_{\npm+1}$ has at least three different values and contradict with the fact that $\vc_{\npm+1}\in\{\threshold_{\npm+1}-1,\threshold_{\npm+1}\}^{\nvote}$. 


Thus, when $\kappa_4$ is not true, we draw the conclusion that
\begin{equation}\nonumber
\forall\,\calH\in\calC,\;\;\dim(\calH_{\leq0}) = \nvote\;\;\text{ when  $\kappa_4$ is false}.
\end{equation}

{\vspace{2mm}\noindent \bf Step 3.2: determine whether the polyhedra is active. }\\
Given the first part of Lemma~\ref{lem:polyhedra}, we know that  $\calH_{\hvPM}\cap\setvote^{\nagent}\neq\emptyset$ is equivalent to the statement of ``$\exists\,\vP\in\setvote^{\nagent}$ such that $r_{\vthreshold,\vbreaking}(\vP)=\hvPM$''. Given the first part of Lemma~\ref{lem:polyhedra}, we know that  $\pi\in\calH_{\hvPM,\,\leq0}$ is equivalent to the statement of ``$\hvPM$ is a refinement of $\pi$''. Combining the above two statements. We know that polyhedra $\calH_{\hvPM}$ is active if and only if $\hvPM$ is an effective refinement of $\pi$. Then, we know that all polyhedra in $\calC$ is not active if and only if $\pi$ has no effective refinements of doctrinal paradoxes. Thus, we know that C2 (and C3) in Lemma~\ref{theo:tech} is equivalent to $\kappa_2$ (and $\kappa_3$) in Theorem~\ref{theo:main}.
\end{proof}

\section{Additional Numerical Results}\label{sec:add_exp}
In Figure~\ref{fig:2}(b), we present the numerical result for the complex setting introduced in Section~\ref{sec:exp}. 
The common settings for the two additional experiments are:\\
\emph{Logical connection}: conclusion $\CL=1$ if $(\PM_1,\PM_2,\PM_3) \in \{(0,0,0),(0,1,0),(1,1,0)\}$ and $\CL=0$ otherwise.\\
\emph{Breaking criteria}: $\vbreaking = (1,0,1,0)$.\\
Other settings and the results are shown in the Table~\ref{tab:exp}.

\begin{table}[ht]
    \centering\vspace{-1.8mm}
    \resizebox{0.55\textwidth}{!}{
    \begin{tabular}{ccc}
    \toprule
     Name of setting & Majority &  Quota\\ \hline
     Threshold $\vthreshold$  & $(0.5,0.5,0.5,0.5)$ & $(0.2,0.2,0.2,0.2)$  \\
     $\pmax$   & $\exp(-\Theta(n))$ & $\Theta(1)$\\
     $\pmin$   & $\exp(-\Theta(n))$ & $\exp(-\Theta(n))$\\
    \bottomrule
    \end{tabular}}
    \vspace{0.3em}
    \caption{Settings and results for the experiments with three premises}\label{tab:exp}
    \vspace{-4.2mm}
\end{table}

\begin{figure*}[ht]
\includegraphics[width = \textwidth]{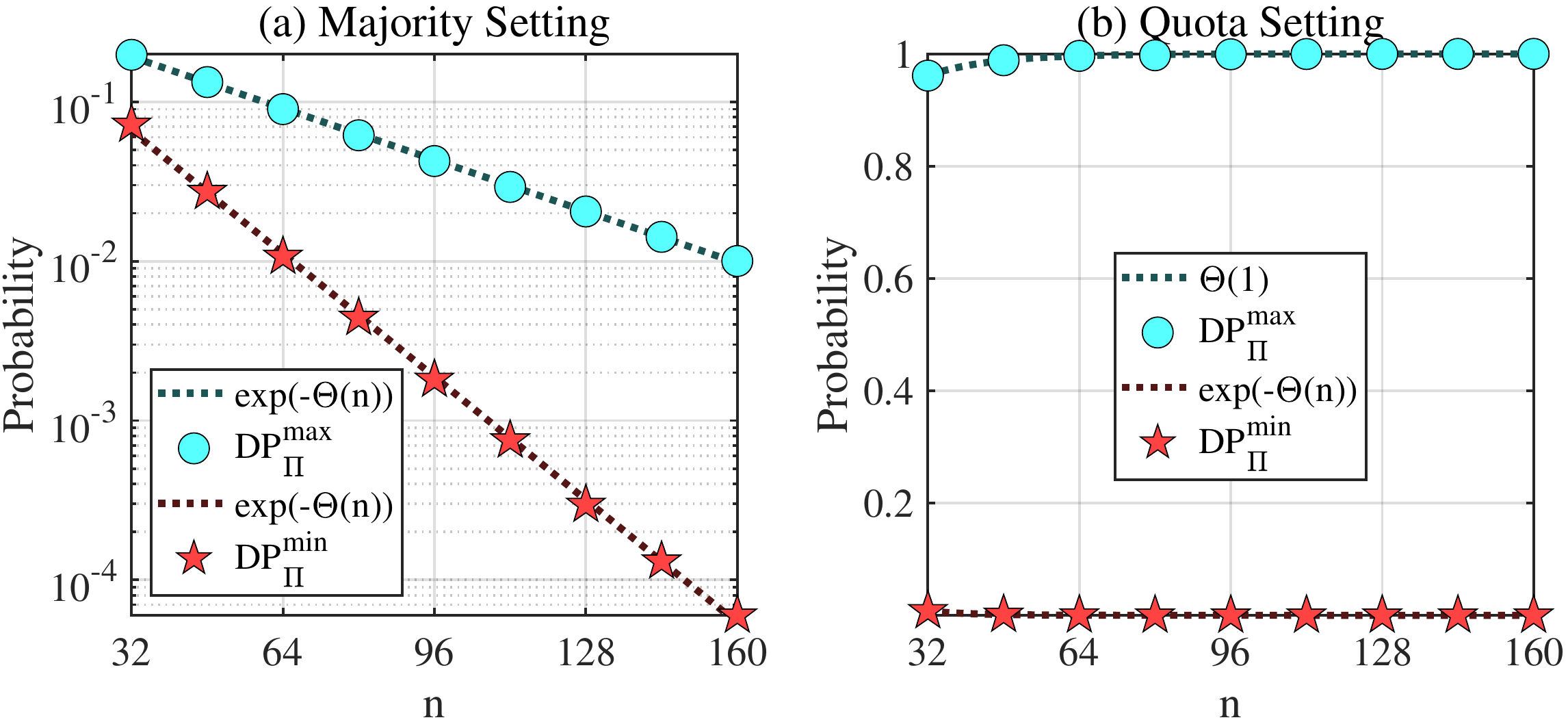}
\caption{\small  A Numerical verification for the complex setting of three premises. Note that (a) is plotted in log-scale.}
\label{fig:2}
\end{figure*}

In Table~\ref{tab:fit1} - Table~\ref{tab:fit4}, we present the expressions of the asymptotic curves under different settings. Root mean square error (RMSE) and coefficient of determination $r^2$ shows the fitness between the asymptotic curves and the data points (simulated smoothed probability of doctrinal paradox).

\begin{table}[ht]
    \centering\vspace{-1.8mm}
    \begin{tabular}{cccc}
    \toprule
     Smoothed Probability & Expressions of the asymptotic curves & RMSE & $r^2$ \\\hline
     $\pmaxg$    & $1 - 0.82227\cdot \exp(-0.05941\cdot n)$ & $9.850\times 10^{-3}$ & $0.99883398$ \\
     $\pming$   & $0.25 - 0.27476\cdot\exp(-0.18796\cdot n)$ & $8.337\times10^{-4}$ & $0.99977090$\\
    \bottomrule
    \end{tabular}
    \vspace{0.3em}
    \caption{The asymptotic curves in  Figure~\ref{fig:1}(a) for Example~\ref{eg:theta1}}\label{tab:fit1}
\end{table}

\begin{table}[ht]
    \centering\vspace{-1.8mm}
    \begin{tabular}{cccc}
    \toprule
     Log-Smoothed Probability & Expressions of the asymptotic curves & RMSE & $r^2$ \\\hline
     $\ln\left(\pmaxg\right)$    & $-0.04904\cdot\nagent-1.0349$ & $3.191\times 10^{-2}$ & $0.99922895$ \\
     $\ln\left(\pming\right)$   &  $-0.09964\cdot\nagent-1.0708$ & $5.838\times10^{-2}$ & $0.99937469$\\
    \bottomrule
    \end{tabular}\vspace{0.3em}
    \caption{The asymptotic curves in  Figure~\ref{fig:1}(b) for Example~\ref{eg:new}}
\end{table}

\begin{table}[ht]
    \centering\vspace{-1.8mm}
    \begin{tabular}{cccc}
    \toprule
     Smoothed Probability & Expressions of the asymptotic curves & RMSE & $r^2$ \\\hline
     $\pmaxg$ when $\nagent$ is even   & $0.96124\cdot\left(n+0.12519\right)^{-1/2}$ & $3.501\times 10^{-3}$ & $0.99949065$ \\
     $\pming$ when $\nagent$ is even  & $0.65588\cdot\exp(-0.64539\cdot n)$ & $7.508\times10^{-4}$ & $0.99989252$\\
    \bottomrule
    \end{tabular}\vspace{0.3em}
    \caption{The asymptotic curves in  Figure~\ref{fig:1}(c) for Example~\ref{eg:1}}
\end{table}

\begin{table}[ht]
    \centering\vspace{-1.8mm}
    \begin{tabular}{cccc}
    \toprule
     Smoothed Probability & Expressions of the asymptotic curves & RMSE & $r^2$ \\\hline
     $\ln\left(\pmaxg\right)$    & $-0.02326\cdot n-0.90777$ & $0.01964$ & $0.99967488$ \\
     $\ln\left(\pming\right)$   & $-0.05555\cdot n-0.94285$ & $0.06556$ & $0.99936553$\\
    \bottomrule
    \end{tabular}\vspace{0.3em}
    \caption{The asymptotic curves in  Figure~\ref{fig:2}(a) for the majority setting of three premises}
\end{table}
\begin{table}[ht]
    \centering\vspace{-1.8mm}
    \begin{tabular}{cccc}
    \toprule
     Smoothed Probability & Expressions of the asymptotic curves & RMSE & $r^2$ \\\hline
     $\pmaxg$    & $1-0.41700\cdot \exp(-0.07436\cdot n)$ & $5.528\times 10^{-4}$ & $0.99833417$ \\
     $\pming$   & $0.25484\cdot\exp(-0.10890\cdot n)$ & $1.5780\times10^{-5}$ & $0.99996718$\\
    \bottomrule
    \end{tabular}\vspace{0.3em}
    \caption{The asymptotic curves in  Figure~\ref{fig:2}(b) for the quota setting of three premises}\label{tab:fit4}
\end{table}

\end{document}